\documentclass{sig-alternate-05-2015}
\makeatletter
\def\ps@headings{%
\def\@oddhead{\mbox{}\scriptsize\rightmark \hfil \thepage}%
\def\@evenhead{\scriptsize\thepage \hfil \leftmark\mbox{}}%
\def\@oddfoot{}%
\def\@evenfoot{}}
\makeatother
\usepackage{algorithmic}
\usepackage{algorithm}
\usepackage{verbatim}

\usepackage{amsmath, amssymb}
\usepackage{enumerate}

\usepackage{eufrak}

\usepackage{verbatim}
\usepackage{array}
\usepackage{graphicx}
\usepackage{float}
\usepackage[footnotesize]{subfigure}

\usepackage{balance}

\newtheorem{theorem}{Theorem}

\newtheorem{corollary}{Corollary}

\newfont{\mycrnotice}{ptmr8t at 7pt}
\newfont{\myconfname}{ptmri8t at 7pt}
\let\crnotice\mycrnotice%
\let\confname\myconfname%

\begin{document}

\title{Content-Centric Networking at Internet Scale \\ 
through The Integration of  Name Resolution and Routing }

\author{J.J. Garcia-Luna-Aceves$^{1,2}$,  Maziar Mirzazad-Barijough$^2$,
Ehsan Hemmati$^2$  \\
$^1$Palo Alto Research Center, Palo Alto, CA 94304 \\
$^2$Department of Computer Engineering,
 University of California, Santa Cruz, CA 95064\\
 Email:  \{jj, maziar, ehsan\}@soe.ucsc.edu }

\CopyrightYear{2016}
\setcopyright{acmcopyright}
\conferenceinfo{ACM-ICN '16,}{September 26--28, 2016, Kyoto, Japan}
\isbn{978-1-4503-4467-8/16/09}
\acmPrice{\$15.00}
\doi{http://dx.doi.org/10.1145/2984356.2984359}

\maketitle

\begin{abstract}


We introduce CCN-RAMP (Routing to Anchors Matching Prefixes), a new approach to content-centric networking.  CCN-RAMP offers all the advantages of the Named Data Networking (NDN) and  Content-Centric Networking (CCNx)
but eliminates the need to either use Pending Interest Tables (PIT)  or lookup large Forwarding Information Bases (FIB) listing name prefixes in order to forward Interests. CCN-RAMP uses small forwarding tables listing anonymous sources of Interests and the locations of name prefixes. Such tables are immune to Interest-flooding attacks and are smaller than the FIBs used to list IP address ranges in the Internet. We show that no forwarding loops can occur with CCN-RAMP, and that  Interests flow over the same routes that NDN and CCNx would maintain using large FIBs.  The results of simulation experiments comparing NDN with CCN-RAMP based on ndnSIM show that  CCN-RAMP requires forwarding state that is orders of magnitude smaller than what NDN requires, and attains even better performance.

\end{abstract}

\vspace{-0.1in}



\section{Introduction}

Several Information-Centric Networking (ICN) architectures have been proposed \cite{icn-survey1, 
icn-survey3, 
icn-survey2} to improve on the  performance of the IP Internet by enabling packet forwarding based on the names of content or services required, rather than the addresses where they may be hosted. They attempt to accomplish this 
by means of new ways to integrate name resolution (mapping of names to locations) and routing (establishing paths between locations) functions.

Section \ref{sec-related} summarizes prior work on  forwarding planes for ICN architectures.
Most architectures keep name resolution and routing  independent of each other, but offer major improvements over the current use of the Domain Name System (DNS) to map domain names to IP addresses.
By contrast,  the Content-Centric Networking (CCNx) \cite{ccnx} and the 
Named Data Networking (NDN) \cite{ndn, ndn-paper} architectures {\em merge} name resolution with routing to content and services in order to allow consumers to request content objects (CO) or services by name.
Routers provide this service using three tables. A content store (CS) lists the COs that are cached locally. A Pending Interest Table (PIT)  keeps forwarding state for each Interest (a request for a CO) processed by a router, such that a single copy of an Interest for the same CO is forwarded and responses to Interests can be sent over the reverse paths traversed by the Interests. A forwarding information base (FIB)  listing the next hops to name prefixes is used to forward Interests towards content producers. 

The simplicity of merging name resolution with routing is very attractive. However, it comes  at a big price, in terms of storage requirements  associated with FIBs and PITs \cite{dai-12,  song15, tsi14, var-13, vir-13, wahl13b}
and additional vulnerabilities to DDoS attacks  associated with PITs \cite{DDos1,  wahl13b}. 


Consider a network in which a name-based routing protocol establishes routes to all known name prefixes and  to those routers that announce name prefixes being locally available, which we call anchors of the prefixes.
Section~\ref{sec-correct} shows that  the path traversed by  an Interest 
when  routers maintain FIBs with entries for name prefixes is the same as the path traversed by the Interest  if the first router binds the name prefix that provides the best match for the CO name to
the name of an anchor of that prefix, 
and routers forward the Interest towards that anchor.

Section~\ref{sec-design} describes CCN-RAMP, which is based on two key innovations.   
First,  
CCN-RAMP 
integrates name resolution with routing to name-prefix anchors 
using the name-based content routing protocol running in the control plane (e.g., \cite{dcr, ehsan-15, nlsr}). The approach takes advantage of the  result presented in Section~\ref{sec-correct}, and  the consequence is 
that routers can forward Interests using forwarding tables that are orders of magnitude smaller than the FIBs required in NDN and CCNx, and even smaller than the FIBs needed in the IP Internet.
Second, 
CCN-RAMP extends recent results in \cite{icnc16, ifip2016} to
eliminate forwarding loops, and to replace PITs with
small forwarding tables listing the next hops towards the origins of Interests, without identifying such origins. 
 
Section~\ref{sec-perf} compares the performance  of CCN-RAMP with NDN when 
either no  caching or on-path  caching is used in a 153-router network. CCN-RAMP  incurs very similar Interest traffic
than NDN to  retrieve content,  requires orders of magnitude fewer entries in forwarding tables than NDN, and needs fewer table look ups to retrieve 
any given CO than NDN.

\section{Related Work }
\label{sec-related}

Previous work attempting to make the forwarding planes of ICN architectures efficient is vast and excellent reviews exist of this prior work \cite{icn-survey1, 
icn-survey3, icn-survey2}. We only highlight the main points regarding name resolution and routing that motivate our design.  In this light, it is important to first note  that a major impediment for  the IP Internet to provide  efficient access to content and services by name is the poor interaction between  name resolution and routing that it currently supports.
A client must be an integral part of name resolution and is required to bind the name of the CO or service to be requested to the location where the content or service is offered. A client first interacts with a local Domain Name System (DNS)  server to obtain the mapping of a domain name to an IP address in order for a request for content or service can be sent to a specific IP address.  Adding to this, the DNS is based on a hierarchical, static caching structure of servers hosting the mappings of domain names to addresses built and maintained independently of routing and requiring servers to be configured on how to contact other servers over the Internet. 

\subsection{ICN Architectures}

TRIAD \cite{triad} was one of the first projects to advocate using names for routing rather than addresses. Since then, many  ICN architectures have been proposed that support name resolution and content routing functionality to enable consumers to ask for content objects (CO) or services by name. These architectures differ on how a CO name is mapped to a producer or source that can provide the requested CO or service, and the way in which paths are established for requests for COs or services and the associated requests.

Some ICN architectures \cite{icn-survey1, icn-survey3, icn-survey2}, including  DONA, PURSUIT, SAIL, COMET, and MobilityFirst,
implement name resolution and routing as independent functions.
In these architectures, name resolution servers (called by different names) are organized  hierarchically,  as multi-level DHTs, or along trees spanning the network \cite{icn-survey3}, and consumers and producers contact such servers to publish and subscribe to content in various ways. Consumers obtain the locations  of publishers from name resolution servers, and  send their content requests to those locations to get the required content or services, and  address-based  routing is used to establish paths between consumers and subscribers or between resolution servers and subscribers or consumers.

A major limitation of  keeping name resolution independent of routing stems from the complexity incurred in keeping name-resolution servers consistent with one another, and allowing consumers and producers to interact with the name-resolution system.  Enabling the updates of name-to-address mapping  is a non-trivial problem using hierarchical structures, spanning trees, or DHT-based organizations of servers. 
Another design consideration in these architectures is that a solution is still needed to preserve the anonymity of the sources of requests for COs or services, which may induce additional complexity in the forwarding plane or require the use of the forwarding mechanisms used in NDN and CCNx.

In contrast to most prior ICN architectures,  NDN and CCNx  {\em merge} name resolution and routing, such that  routers are the facto name resolvers by establishing  routes to name prefixes on a hop-by-hop basis. A major advantage of doing this is that it eliminates the complexity of designing and maintaining a  network of name-resolution servers that replace the DNS.  This merging of functionalities  is supported by: (a) a name-based routing protocol operating in the control plane, which updates the entries in FIBs listing the next hops to known name prefixes, and (b) forwarding of Interests based on the longest prefix match (LPM)  between the CO name  in the Interest and a name prefix listed in the FIBs.

However, as attractive as the simplicity of merging of name resolution with routing in NDN and CCNx is, it comes at a very big price.  Because the name of a CO or service is bound directly to a route on a hop-by-hop basis, each router along the path traversed by an Interest must look up a FIB listing the known name prefixes.  To operate at Internet scale, FIB sizes in NDN 
are acknowledged to eventually reach billions of entries \cite{song15}. This is easy to imagine, given that the number of registered domains in the IP Internet was more than 300 million by the end of 2015. 
This is orders of magnitude larger than the largest FIB size for the IP Internet, which is smaller than 600,000 today. 

To make matters worse, the name prefixes assumed in ICN architectures are  variable length and much more complex than IP addresses. This means that efficient LPM  algorithms developed for the IP Internet cannot be applied directly to NDN and CCNx. Indeed,  it has been noted that NDN and CCNx cannot be deployed at Internet scale without further  advances in technology \cite{peri-11}.

\vspace{-0.05in}
\subsection{Limitations of Using PITs }

The size of a PIT grows linearly with the number of distinct Interests received by a router as consumers pipeline Interests (e.g., to support HD video streams) or request more content, or more consumers request content \cite{dai-12, tsi14, var-13}.   Unfortunately, as the following paragraphs
summarize, 
PITs  do not deliver substantial benefits compared to much simpler Interest-forwarding mechanisms, and can actually be counter-productive. 

We have shown \cite{ali-ifip16, ifip2016} that the percentage of 
aggregated Interests  is minuscule when in-network caching is used, even  Interests exhibit temporal correlation.  We have also shown that per-Interest forwarding state is not needed to  preserve the privacy of consumers issuing the Interests  \cite{icnc16, ifip2016}. 

Supporting  multicast content delivery efficiently in the data plane has been viewed as a major reason to use PITs. However, as we demonstrate in \cite{globe2016}, maintaining per-Interest forwarding state is unnecessary to implement pull-based multicast content dissemination. In a nutshell,  a source-pacing algorithm can be 
used with routers maintaining per-source rather than per-Interest forwarding state. Routers forwarding traffic from a given multicast source 
maintain the most recent {\em multicast-counter value}  ($mv$) for the source.
A router forwards an Interest towards a source only if the $mv$ stated in the Interest is larger than the value it currently stores for the source and discards the Interest otherwise. The net effect is the same as Interest aggregation, but without the large overhead of per-Interest forwarding state \cite{globe2016}.

We  have also shown  \cite{ifip2015,  ancs2015, nof2015} that Interest aggregation combined with the Interest-loop detection mechanisms used in NDN and CCNx can lead to Interests being aggregated while traversing  forwarding loops without such loops being detected. This results in aggregated Interests ``waiting to infinity" for responses that never come. 
In addition, using PITs makes routers vulnerable to  Interest-flooding attacks \cite{DDos1, vir-13, wahl13a, wahl13b} in which malicious users can send malicious Interests aimed at  making the size of PITs explode. 
Unfortunately,  the countermeasures that have been proposed  for these attacks \cite{afan} simply attempt to reduce the rates at which suspected routers can forward  Interests, and this can be used to mount other types of denial-of-service attacks.

\subsection{Limitations of Using FIBs \\ Listing Name Prefixes }

A number of proposals have been advanced to allow Interest forwarding in NDN based on LPM to keep up with new wire speeds while coping with the required FIB sizes and name-prefix structures.  
A big challenge for name-based Interest forwarding is to attain 100 Gbps rates  or higher, given that FIBs listing name prefixes  at Internet scale are much too large to fit into SRAM or TCAM \cite{peri-11}.

Several  proposals for the implementation of FIBs  for NDN rely on  the use of tries and massive parallelism in order to avoid bottlenecks in the encoding process needed to use the tries \cite{wang12, wang13}. Other approaches are based on hash tables for FIB lookups \cite{so12, so13, var-12, yuan-12},  which requires  larger memory footprints and is not scalable to prefix names with a large number of name components. Hash-based approaches for name-based forwarding are based on  DRAM technology and rely on massive parallel processing, because they would require  hundreds of MiBs for just a few million prefixes. 

Given the major limitations of  using  LPM in FIBs containing billions of name prefixes, a few proposals have been advanced to either reduce the size of FIBs 
listing name prefixes or eliminate the use of such FIBs.   
 
Song et al. \cite{song15} introduced the concept of ``speculative forwarding" based on longest-prefix classification (LPC) rather than LPM. LPC behaves just like LPM when 
a match is found in the FIB for the name stated in an Interest; however, with LPC a packet is forwarded to a next hop given by the FIB even if no match is found.  
Unfortunately,  as described in  \cite{song15}, Interests may be forwarded along loops.  Although the NDN forwarding strategy can prevent Interests from traversing the same forwarding loop multiple times, it cannot guarantee that Interests will not be aggregated while they traverse forwarding loops \cite{ancs2015}.

SNAMP \cite{afan15}  and PANINI  \cite{panini16} reduce  FIB sizes by means of default routes and default-free zones. 
Edge routers resolve and  keep track of local name prefixes, and forward all  other interests toward backbone routers
that  create a default-free zone and map name prefixes to globally-routed names. 

TagNet \cite{papa-14} uses content descriptors and host locators for forwarding.
Content descriptors are variable-length sets represented with fixed-length Bloom filters for forwarding. 
Host locators are routing labels assigned to routers and hosts along one or multiple spanning trees. The 
labeling  approach used in TagNet to assign locators is due to  Thorup and Zwick \cite{thoru-01}. TagNet results in FIBs that are much smaller than the FIBs required in NDN; however, it has a number of limitations. Content requests (Interests) 
must state the locators of  their sources, which eliminates  the anonymity  provided in NDN and CCNx. Like any other scheme based on compact-routing, 
the paths traversed over the labeled spanning trees can have some stretch over the shortest paths. The 
Thorup-Zwick labeling used in TagNet is a depth-first search 
approach, and entire spanning trees may have to be relabeled after a single link failure. No prior work exists showing that TagNet or similar routing schemes based on interval-routing labels are suitable for large networks subject to topology changes or mobility of hosts and routers.

\section{Routing to Name prefixes \\ vs. Routing to Anchors}
\label{sec-correct}

We use the term {\em anchor} to denote a router that, as part of the operation of the name-based routing protocol, announces  the content corresponding to a name prefix being locally available. 
If multiple mirroring sites host the content corresponding to a name prefix, then the routers attached to those sites announce the same name prefix. However,  a  router simply caching COs from a name prefix does not announce the name prefix in the name-based routing protocol.  

An anchor announcement can be done implicitly or explicitly as part of the operation of the name-based routing protocol running in the control plane, and routers simply caching COs are not anchors. For example, in NDN \cite{nlsr} and other name-based routing protocols  based on link-state information \cite{ehsan-15}, routers exchange link-state announcements (LSA)  corresponding to either name prefixes or  adjacencies to networks or routers. In this context, any router that
originates an LSA for a name prefix  is an anchor of the  prefix. On the other hand, DCR  \cite{dcr} and other  name-based content routing protocols  (e.g., \cite{gold}) based on distance information to name prefixes  use the names of anchors to ensure the correctness of the multi-path route computation \cite{icnp14}.

Figure \ref{fig:fibs} shows an example of a content-centric network in which router  $y$ is an anchor  for prefixes $P^*$, $Q^*$, and $R^*$; and router $z$ is an anchor for prefixes $A^*$, $C^*$, and $P^*$. Dark solid arrowheads show the best next hop towards name prefixes. A red arrowhead indicates the best next hop to name prefixes with COs locally available at router $y$, and a dashed arrow head indicates the best next hop towards name prefixes with COs locally available at router $z$.  The FIB entries at router $p$ are shown when FIB entries are maintained for all instances of each name prefix, the nearest instances of a name prefix, or only for the anchors of name prefixes.

 \vspace{-0.18in}
\begin{figure}[h]
\begin{centering}
    \mbox{
    \subfigure{\scalebox{.21}{\includegraphics{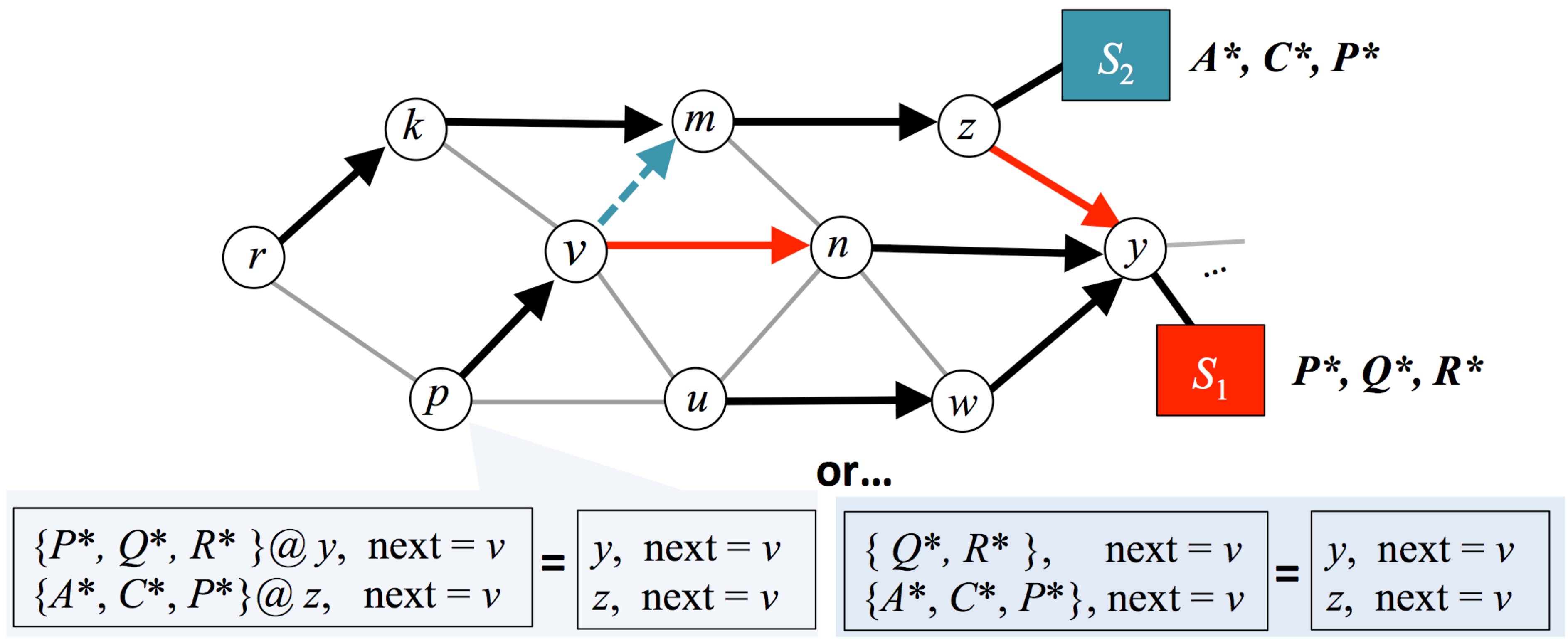}}}
      }
  \vspace{-0.1in}
   \caption{FIB entries for name prefix instances and FIB entries for  anchors of
    name prefix instances
   }
   \label{fig:fibs}
\end{centering}
\end{figure}

This example  illustrates  the fact that routes to instances of name prefixes must also be routes to the anchors announcing those instances. Therefore,
the paths obtained from FIB entries listing name prefixes are the same as the paths obtained from FIB entries listing the anchors of name prefixes.

It is important to note that a router 
acts as the anchor of a name prefix over time scales that are many orders of magnitude larger than  either the time scale at which congestion varies in a network due to traffic or even topology changes, or the time needed for all routers to respond to congestion changes. 

Even though current Internet  routing protocols do not handle congestion or multi-path routing well, 
the necessary  and sufficient conditions for minimum-delay routing (or minimum-congestion routing) are well known \cite{gallager} and practical approaches that provide very good approximations to optimum routing using datagrams  have existed for some time \cite{murthy, vutukury}. 
Given the long  time periods over which routers act as anchors of name prefixes, 
these approaches
can be used in either routing to name prefixes or  routing to the anchors of name prefixes. Furthermore,  existing name-based routing protocols can easily 
apply mechanisms to react to congestion, namely: multi-path routing,  maintaining congestion information about local interfaces, and  path-based measurements of congestion. Accordingly, it  can be safely assumed that congestion-oriented multi-path routing is attained  independently of whether routes are established for name prefixes or anchors.

The following theorem formalizes the result illustrated in Figure 1 for the case of shortest-path routing (single path or multi-path) to name prefixes. We show  
that, if the same  name-based routing protocol is used to establish routes to anchors and to name prefixes,  the paths traversed by Interests are the same or mostly the same.

\begin{theorem}
\label{theo1}
The paths to name prefixes obtained using forwarding entries listing the name prefixes are the same as the paths obtained using forwarding entries listing the anchors of name prefixes in a stable  network in which a correct name-based routing protocol is executed.
\end{theorem}

 \begin{proof}
Consider a content-centric network in which forwarding entries list the next hops to the anchors of name prefixes. Assume that the routing protocol  computes the paths to all the known anchors by time $t_0$ and that no changes occur in the network after that time.

Consider a name prefix $P$ for which router $a$ is an anchor, and assume for the sake of contradiction that, at some time $t_1 > t_0$, there is a route from a router $r$ to the instance of prefix $P$ announced by anchor $a$  that is shorter than 
any of the  routes from $r$ to $a$ implied by the forwarding tables maintained by routers at 
$t_1$.   This is a contradiction to the definition of an anchor and the assumption that  the routing protocol  
computes correct routes to all anchors by time $t_0$ and no changes occur after that. \end{proof}

\begin{corollary}
\label{coro}
For any name prefix that has a single anchor, the paths to the  name prefix obtained using forwarding entries listing  name prefixes 
are the same as the paths obtained using forwarding entries listing the anchors of name prefixes.
\end{corollary}

 \begin{proof}
The proof follows from Theorem 1 for the case of a stable topology. If
forwarding tables are inconsistent due to network dynamics, the result follows  from the one-to-one correspondence between a name prefix and its anchor, given that  the prefix is hosted at a single site.
 \end{proof}

Figure \ref{fig:ex2} illustrates a content-centric network in which name prefix $P^*$ is multi-homed. Assume that router $o$ receives an Interest for a CO with a name in prefix $P^*$ from consumer $c$, and the nearest anchor of $P^*$ to router $o$ is $a_i$.  In the example, $D_{ij}$ denotes the length of the pat from router $i$ to router $j$, and $D^*_{r a_i}$ denotes the distance from $r$ to $a_i$ along a path that does not include $p$.

If routing is based on anchor names, router $o$ binds the CO name to anchor $a_i$. Assume that the shortest path from $o$ to $a_i$ 
includes router $r$ and that $p$ is the next hop from $r$ to $a_i$.
Consider the case in which, because   link $(r, p)$ fails or becomes too congested, router $r$ must find an alternative path to forward  
an Interest intended for  anchor $a_i$. Anchor $a_j$ is the closest anchor of prefix $P^*$ to router $r$ after the change in link $(r, p)$. Routing based on name prefixes can provide more efficient forwarding than routing based on anchor names only if one of the following  conditions is true:

{\small
\begin{equation}
D^*_{r a_i} < \infty ~ \wedge ~ 
D_{r a_j} < D^*_{r a_i} 
\end{equation}
\begin{equation}
D^*_{r a_i} = \infty ~ \wedge ~ 
D_{r a_j}  < D_{r o} + D_{o a_j}
\end{equation}
}

 \vspace{-0.05in}
Only a small number of Interests pertaining to name prefixes that are multi-homed may be forwarded more efficiently if routing based on name prefixes is used.  The reason for this is that Equations 1 and 2 cannot be satisfied over extended time periods. The name-based routing protocol operating in the control plane must make router $o$ update 
its distances to anchors reflecting the change in  link $(r, p)$.  Once router $o$
updates its forwarding table, it must select $a_i$  or another anchor as the new nearest anchor for $P^*$ and Interests would
traverse  new shortest paths. We observe that this is the case even if mobility of hosting sites or consumers occurs, because the efficiency of Interest forwarding is determined by the distances between the routers attached to consumers and the routers attached to hosting sites (anchors).

 \vspace{-0.10in}
\begin{figure}[h]
\begin{centering}
    \mbox{
    \subfigure{\scalebox{.18}{\includegraphics{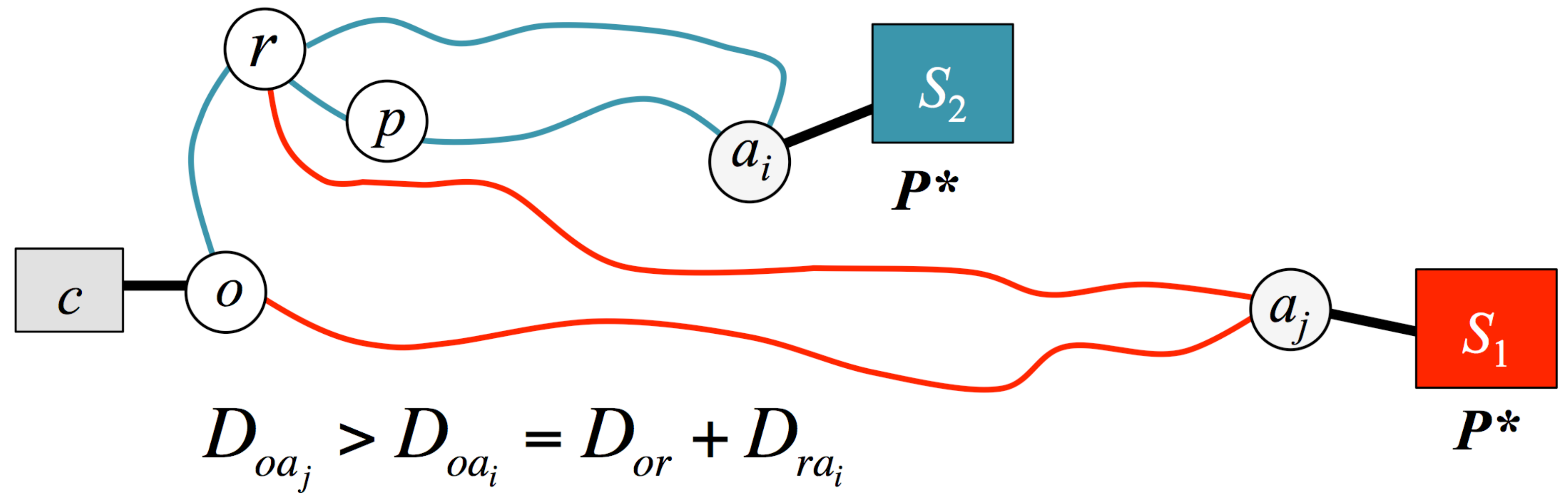}}}
      }
  \vspace{-0.1in}
   \caption{Forwarding Interests based on routes to anchors or routes to name prefixes
   }
   \label{fig:ex2}
\end{centering}
\end{figure}

\section{CCN-RAMP }
\label{sec-design}

\subsection{Design Overview}

CCN-RAMP allows routers to forward Interests by looking up small  forwarding tables listing next hops to anchors. 
Clearly, for this  to work, routers need to first 
obtain the mapping of  the CO name stated in an Interest to an anchor 
of the name prefix that best matches the CO name. 

CCN-RAMP takes advantage of the fact that 
existing name-based routing protocols communicate the anchor of a name prefix as part of an LSA  or a distance update for the prefix \cite{dcr, ehsan-15, nlsr}.  
Using the information disseminated in the name-based routing protocol, a router builds and maintains two tables:  A Forwarding to Anchors Base (FAB) listing the routes to anchors, and a Prefix Resolution Table (PRT)  listing the anchors of each name prefix. 

A router receiving an Interest from a local consumer  (call it origin router) uses its PRT to bind the CO name to the nearest  anchor for the name prefix that is the best match for the CO name. To allow relaying routers to use only their FABs to forward Interests,  an Interest states the name of the anchor chosen by the origin router. The origin router and other relaying routers  forward the Interest as needed using their FABs and the anchor name in the Interest.  

Routers in CCN-RAMP use exactly the same amount of routing signaling  as routers in NDN and CCNx, and  a PRT in CCN-GRAM has as many entries as a FIB in NDN and CCNx. However, the amount of  routing signaling is not a problem, because anchors of name prefixes  change infrequently and the name-based routing protocol can send updates regarding anchor-prefix bindings independently of updates regarding how to reach anchors.
The limitation of using a FIB listing name prefixes is not its size but the need to look up the FIB in real time for each Interest being forwarded, which either requires  very expensive memory or renders very slow lookup times. 

In NDN and CCNx, each router forwarding an Interest must look up a FIB based on name prefixes  because name resolution and routing are merged into one operation. By contrast,  name resolution in CCN-RAMP is delegated to the origin routers that receive Interests from local consumers. 
Each router is in effect a name resolver. 
An origin router that binds a CO name to an anchor by looking up its  PRT   in response to an Interest from a local consumer replaces the role of the DNS to resolve a name into an address. However, consumers are relieved from being involved in name resolution. 
What can take hundreds of milliseconds using the DNS takes just a lookup of the PRT, which can be implemented using tries and slow storage. Once the binding of a CO name to an anchor is done by an origin router, 
forwarding an Interest involves fast lookups of FABs that can be even smaller than the FIBs used in the IP Internet today because only anchors are listed. 

The tradeoff made in  CCN-RAMP compared to prior approaches that separate name resolution from address-based routing is the need for each router to store a PRT. This is acceptable, given the cost of memory today and the infrequency with which PRT entries must be updated.

The design of CCN-RAMP eliminates forwarding loops by ordering the routers  forwarding  Interests based on their distances to destinations  \cite{ancs2015}.
To attain this, each Interest carries the distance to an anchor and FABs list the next hops {\em and} the  distances to anchors.

CCN-RAMP extends our prior work  \cite{icnc16, ifip2016}  to eliminate the use of PITs while  providing the same degree of Interest anonymity enabled in NDN and CCNx. 
A router maintains a Label Swapping with Anchors Table (LSAT) 
to remember the reverse paths traversed by Interests, and uses anonymous identifiers (AID) with local scope to denote the origins of Interests. 
The origin of an Interest is denoted with an AID that is swapped at each hop.

\subsection{Assumptions}

We make a few assumptions to simplify our  description of CCN-RAMP; however, they  should not be considered  design requirements. For convenience, a request for  content  from a local user is sent to its local router in the form of an Interest. 

Interests are retransmitted only by the consumers that originated them, rather than routers that relay Interests,  routers forward Interest based on LPM, and a router  can determine whether or not the CO with the exact same CO name is stored locally. Routers know which interfaces are neighbor routers and which are local consumers, and forward  Interests on a best-effort basis.  

Allowing anchors with only subsets of the COs in the name prefixes they announce 
requires routers to determine which of the anchors announcing the name prefix actually host the requested CO.
This can be accomplished in CCN-RAMP using  the
multi-instantiated destination spanning trees (MIDST) described in \cite{dcr, icnp14}. 
The pros and cons of allowing anchors to host only subsets of the COs in name prefixes, and the search mechanisms needed to support it,  are the subject of another publication. In the rest of this paper, we assume that each anchor of a name prefix is required to
have   {\em all} the COs in the name prefix locally available.

The name of content object (CO)  $j$ is denoted by  $n(j)$
and the name prefix corresponding to the longest prefix match  for  name $n(j)$ is denoted by $n(j)^*$.
The set of neighbors of router $i$ is denoted by $ N^i$.

\subsection{Information Exchanged }
\label{sec-info}

An Interest forwarded  by router $k$ requesting CO with name $n(j)$ is denoted by  $I[n(j), AID^I(k), 
a^I(k), D^I(k) ]$, and states the name  of the CO ($n(j)$),  a fixed-length 
anonymous identifier ($AID^I(k)$) denoting the 
origin router of the Interest, the anchor selected by the first router processing the Interest 
($a^I(k)$), and the distance ($D^I(k)$) from $k$ to  $a^I(k)$.
 
A data packet sent by router $i$ in response to  an Interest
is denoted by  $DP[n(j), AID^R(i), sp(j) ]$ and  states, in addition to a CO, the name  of the  CO being sent ($n(j)$),  an anonymous identifier ($AID^R(i)$) denoting 
the router that should receive the data packet, and a  security payload ($sp(j)$) used optionally to validate the CO.

An error message sent by router $i$ in response to  an Interest is denoted by $ERR[n(j),$   $ AID^R(i)$, $a^R(i)$,  $ \mathsf{CODE} ]$  and states the name  of a CO ($n(j)$),   an anonymous identifier ($AID^R(i)$) that states the intended recipient of the reply, 
the selected anchor for the name prefix ($a^R(i)$), and a code ($\mathsf{CODE}$) indicating the reason why the reply is sent. Possible reasons for sending a reply include: an Interest loop is detected, no route is found towards requested content, no content is found, and an upstream link is broken.

\subsection{Information Stored}
\label{sec-info2}

Router $i$ maintains three tables for packet forwarding:  A
Prefix Resolution Table ($PRT^i$), a Forwarding to Anchors Base ($FAB^i$),  and a Label Swapping with Anchors Table ($LSAT^i$).
If router $i$ has local consumers, it  maintains 
a Local Request Table ($LRT^i$). Router $i$ maintains  a Content Store ($CS^i$) if it provides content caching locally.

$PRT^i$ is indexed by the known name prefixes advertised by their anchors. Each entry of the $PRT^i$ states the names of the selected anchors that advertised the prefix. Depending on the specific approach, the list may state the nearest anchors or all the anchors of the name prefix. However, with the assumption that anchors must have all COs of the name prefixes they announce, listing the nearest anchors for a name prefix suffices.

$FAB^i$ is indexed by anchor names and each entry in $FAB^i$ states available next hops to the anchor. The distance stored for  
neighbor $q$ for anchor $a$  in $FAB^i$  is denoted by $D(i, a, q)$. This information is updated by means of a name-based routing protocol running in the control plane.

$LSAT^i$ is  indexed by anonymous identifiers denoting origin routers. An anonymous identifier (AID) is simply a fixed-length number.
Each  entry in $LSAT^i$ states an 
$AID$ locally created or received in Interests from a previous hop,  
the previous hop  ($PH^i[AID]$) that provided the AID, a next hop ($NH^i[AID]$), 
the mapped AID  ($MAP^i [AID]$) that should be used for the next hop, and  the distance 
($D^i[AID]$) to the anchor that should receive the forwarded Interests.

$LRT^i$ lists the names of the COs requested by router $i$ on behalf of local consumers. The entry for CO name $n(j)$  states the name of the CO ($n(j)$) and a list of the identifiers of local consumers (denoted by $lc[n(j)]$) that have requested the CO.
$CS^i$ lists the COs cached locally.  The entry for CO name $n(j)$  states a pointer to  the content of the CO (denoted by $p[n(j)]$).

\subsection{Avoiding Forwarding Loops}
\label{lfr}

Let $S^i_a$ denote the set of next-hop neighbors of router $i$ for anchor $a$. 
Router $i$ uses the following  rule to 
ensure that Interests cannot traverse forwarding loops, even if the 
forwarding data maintained by routers regarding name prefixes and anchors is inconsistent or contains  routing-table loops.

 \vspace{0.05in}
\noindent
{\bf  Anchor-Based Loop-Free Forwarding  (ALF):}  \\
If router $i$ receives $I[n(j), AID^I(k), a^I(k) = a, D^I(k) ]$ from router $k$,
it can forward  $I[n(j), AID^I(i), a^I(i) = a, D^I(i) ]$  if:

\vspace{-0.07in}
\begin{enumerate}
\item
$AID^I(k) \not\in LSAT^i \wedge ~ \exists v \in  S^i_{a} ~(~ D^I(k)  > D(i, a, v) ~) $
\vspace{-0.05in}
\item
$AID^I(k) \in LSAT^i  ~\wedge~   (~D^I(k)  > D^i[AID^I(k)] ~) $
\end{enumerate}

\begin{theorem}
\label{theo}
No Interest  can traverse a forwarding loop in a content-centric network in which  CCN-RAMP is used.
\end{theorem}

 \begin{proof}
Consider a network in which CCN-RAMP is used and  assume that, following the operation of CCN-RAMP, there is a router $v_0$ that originates an Interest for CO $n(j)$, uses longest match prefix to obtain the name prefix $n(j)^*$, and binds that name prefix to anchor $a$. The Interest sent by $v_0$ is
$I[n(j), AID^I(v_0), a^I(v_0) = a, D^I(v_0) ]$.

For the sake of contradiction,  assume that  routers in a forwarding  loop  $L$ of $h$ hops  $\{ v_1 , $ $v_2 , ..., $ $v_h , v_1 \}$  forward the Interest   for CO $n(j)$ originated by $v_0$ along $L$, with no router in $L$  detecting that  the Interest  has traversed  loop $L$. 

Given that  $L$ exists by assumption, router $v_k \in L$ must forward  $I[n(j),  AID^I(v_k), a^I(v_k) = a, D^I(v_k)]$ to router $v_{k+1}$ $ \in L$ for $1 \leq k \leq h - 1$, and router $v_h \in L$ must forward $I[n(j),$   $ AID^I(v_{h}), $ $a^I(v_h) =$ $ a, D^I(v_{h}) ]$ to router $v_{1} \in L$. 
According to ALF, if router $v_k$ forwards Interest $I[n(j), $   $AID^I(v_{k}), a, D^I(v_{k}) ]$ to router $v_{k+1}$ as a result of receiving  $I[n(j),$  $AID^I(v_{k - 1}),$ $ a, D^I(v_{k - 1}) ]$ from router $v_{k-1}$, then  it must be true that  

\vspace{-0.14in}
{\small
\[ 
AID^I(v_{k - 1}) \not\in LSAT^{v_k} 
\wedge [~D^I(v_{k-1})  > D(v_k, a, v_{k+1}) ]
\]
}

\vspace{-0.2in}
\noindent
or

\vspace{-0.2in}
{\small
\[
 AID^I(v_{k - 1}) \in LSAT^{v_k} 
\wedge [~D^I(v_{k-1})  > D^{v_k} [AID^I(v_{k - 1})~] . 
\]
}

 \vspace{-0.12in}
Similarly, if router $v_1$ 
forwards Interest $I[n(j), AID^I(v_{1}), $ $a, D^I(v_{1}) ]$ to router $v_{2}$ as a result of receiving  $I[n(j),$  $AID^I(v_{h}),$ $ a, $ $D^I(v_{h}) ]$ from router $v_{h}$, then 

{\small
\[ 
AID^I(v_{h}) \not\in LSAT^{v_1} \wedge [~D^I(v_{h})  > D(v_1, a, v_{2}) ]
\]
}

\vspace{-0.2in}
 \noindent
 or

\vspace{-0.2in}
{\small
\[
AID^I(v_{h}) \in LSAT^{v_1} \wedge  [~D^I(v_{h})  > D^{v_1}  [AID^I(v_{h}) ]  ~].
\]

}

\vspace{-0.03in}
Given that each  router in loop $L$ that forwards  an Interest for a given AID for the first time must create an entry in its LSAT, 
it follows from the above argument that, for loop $L$ to exist and be undetected when each router in  the loop uses ALF to forward the Interest  originated by router $v_0$, it must be true that  
\begin{equation}
D^I(v_{k-1})  > D^{v_k} [AID^I(v_{k - 1})]  ~for~ 1 < k \leq h
\end{equation}

\vspace{-0.2in}
\begin{equation}
D^I(v_{h})  > D^{v_1} [AID^I(v_{h})] .
\end{equation}

However, Eqs. (3) and (4) constitute a contradiction, because they imply that  $D^I(v_{k}) > D^I(v_{k})$ for $1 \leq k \leq h $. Therefore, the theorem is true. 
\end{proof}

Theorem 2 is  independent of whether the network is static or dynamic, the specific caching strategy used in the network, the retransmission strategy used by content consumers  or relay routers after experiencing  a timeout or receiving a reply, or whether routers use multiple paths or a single path to forward Interests towards a given anchor.
We should also point out that ALF is a sufficient condition to ensure loop-free Interest forwarding, and it is possible that more flexible loop-free forwarding rules
could be found. This is the subject of future work.

\subsection{ Interest Forwarding }
 \label{sec-aid} 

Figure \ref{fig:fw} illustrates the forwarding of  Interests  in CCN-RAMP  showing two anchors ($y$ and $z$) and five name prefixes, one of them ($P^*$) is multi-homed at $y$ and $z$. The figure shows the forwarding tables used at router $k$ (PRT, FAB and LSAT). Not shown are the content store and the LRT at router $k$.

 \vspace{-0.10in}
\begin{figure}[h]
\begin{centering}
    \mbox{
    \subfigure{\scalebox{.19}{\includegraphics{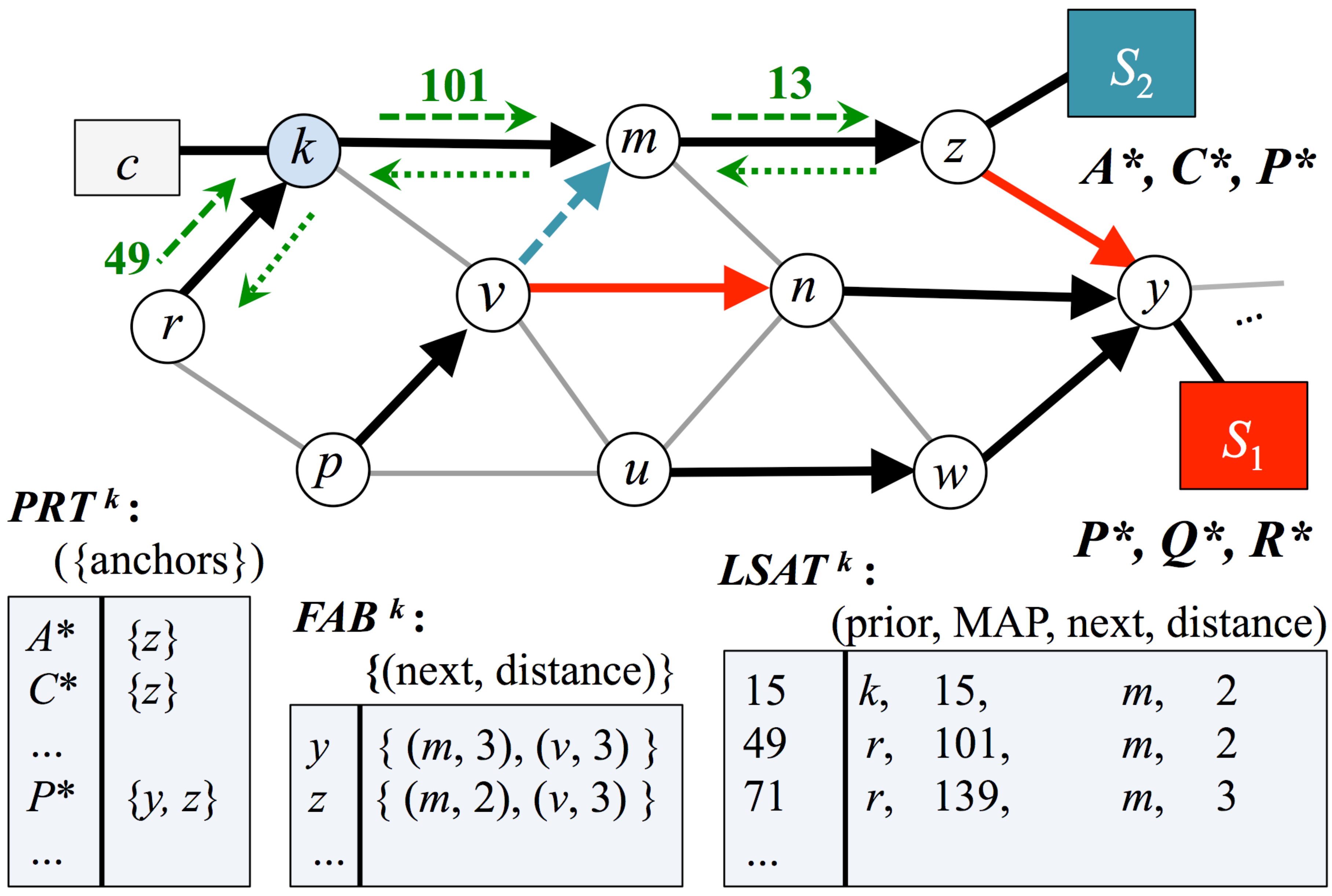}}}
      }
  \vspace{-0.1in}
   \caption{Interest forwarding in CCN-RAMP
   }
   \label{fig:fw}
\end{centering}
\end{figure} 

 \vspace{-0.05in}
When router $k$ receives an Interest from a local consumer $c$ for a CO with name $n(j)$, it looks up its content store ($CS^i$) to determine if the CO is stored locally.  If the CO is remote, router $k$ adds $c$ to an entry in $LRT^k$ stating the list of consumers that have requested $n(j)$ and proceeds to create Interest 
$I[n(j), $ $AID^I(k), $ $a^I(k),$ $ D^I(k)]$. The LRT entries allow router $k$ to demultiplex responses it receives 
for AIDs that it originates   and send the responses  to the correct consumers.

Router $k$ looks up  $PRT^k$ for the name prefix that provides the best match and selects an anchor of that prefix to be included in the Interest ($a^I(k)$). The router then looks up $FAB^k$ to obtain the 
next hop $m$ and the distance to the selected anchor. 
Router $k$ includes its distance to the anchor ($D^I(k)$) in its Interest, so that  forwarding routers can apply ALF as described in Section 4.5. 
If no entry exists in $LSAT^k$ with 
Router $k$  as the prior hop, it selects 
an anonymous identifier ($AID^I(k)$) to denote itself as the origin of the Interest, such that the identifier is  not being used by $k$ to denote any other origin of Interests in $LSAT^k$.  To select new AIDs, router $i$  maintains a hash table or an array of bits that keeps track of previously used random numbers. An alternative approach could be using a counter that is increased after creating a new 
mapped AID (MAP).

In the example of Figure \ref{fig:fw}, router $k$ has forwarded an Interest from a local consumer and used 15 as the  anonymous identifier (AID) to identify itself. Router $k$ can use the same AID in all Interests it sends towards any anchor on behalf of local consumers, or use different AIDs.

If router $k$ forwards Interest $I[n(j), $ $AID^I(r), $ $a^I(r),$ $ D^I(r)]$ from neighbor $r$ to neighbor $m$,  ALF is satisfied, 
and no entry for $AID^I(r)$ exists in $LSAT^k$, then 
router $k$  computes an AID   that is not used as the MAP  in any entry in $LSAT^k$. Router $k$ then creates the  entry for $AID^I(r)$ in $LSAT^k$ stating: $PH^k(AID^I(r)) = r$, $NH^k(AID^I(r))=$ $ m$, $MAP^k$ $(AID^I(r))=$ $ n$, and  $D^k(AID^I(r))=$ $ D(k, a^I(r), m)$.

Note that only the ingress router receiving an Interest from a local consumer
(e.g., router $k$ receiving an Interest from $c$ in the example) needs to look up its PRT, which is of the same size as a FIB in NDN and CCNx. 

Forwarding routers  use only their FABs and LSATs.
Furthermore, a router needs to lookup  its $FAB$ only when no forwarding state exists in  its $LSAT$ for the AID given in an Interest received from a neighbor router.
A forwarding router receiving an Interest looks up its LSAT. If forwarding state is already set in its LSAT for the AID stated in an Interest from a given neighbor, the router forwards the Interest without involving its FAB. 

The forwarding state in $LSAT^k$ specifies the next hop and AID to be used to forward an Interest received with a given AID from a previous hop.  In the example,  router $k$ maps $(AID = 49, prior = r)$ in the Interests received from $r$  to $(AID = 101, next = m)$ in the Interests it forwards to $m$ towards anchor $z$.  The dashed green arrows in Figure~\ref{fig:fw} show the flow of Interests from $r$ to anchor $z$ and the flow of responses from $z$ to $r$. 

The AID swapping approach adopted in CCN-RAMP  simplifies the approaches we introduced in  \cite{icnc16, ifip2016}. Its intent is to make  Interest forwarding as simple and fast as label swapping in IP networks, without revealing the identities of the consumers or routers that originate  the Interests.

 \vspace{-0.07in}
\begin{figure}[h]
\begin{centering}
    \mbox{
    \subfigure{\scalebox{.19}{\includegraphics{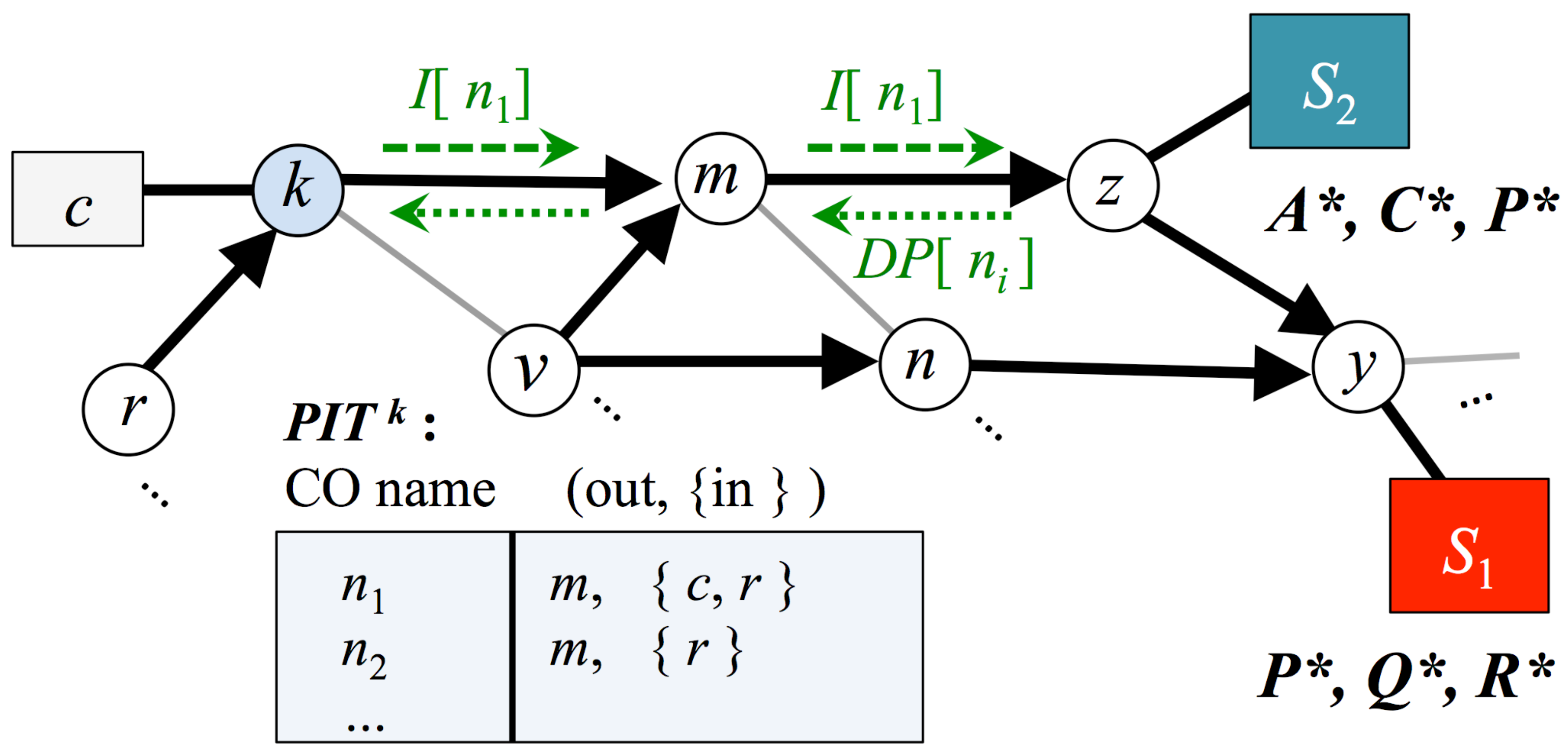}}}
      }
  \vspace{-0.1in}
   \caption{Interest forwarding in NDN 
   }
   \label{fig:pit}
\end{centering}
\end{figure} 

 \vspace{-0.07in}
Figure~\ref{fig:pit} helps to illustrate the level of anonymity provided in NDN and CCNx. Interests in NDN only state the name of the requested CO and a nonce, and the per-Interest forwarding state stored in PITs is what allows routers to forward responses to Interests over the reverse paths traversed by the Interests.
A third party monitoring traffic cannot determine the origin of an Interest
simply from the information in the headers of Interests. However, the 
origin of an Interest can be  obtained if routers collaborate and trace back the path the Interest traversed using the PIT entries listing the CO name in the Interest.

The same type of anonymity is provided in CCN-RAMP, but without the need for per-Interest forwarding state. A third party monitoring traffic cannot determine the source of an Interest simply by reading the information in the header of the Interest, because a local identifier is used to denote the source of an Interest at each hop. However,  routers can collaborate to trace back the origin of Interests  by means of the LSAT entries
stored by the routers.

\subsection{Updating Forwarding State }
 \label{sec-labeling}

Algorithms 1 to 4 show the steps taken by routers to maintain the forwarding state needed to forward Interests, COs and error messages.  The algorithms assume that $FAB^i$ and $PRT^i$ are initialized and maintained  by a routing protocol operating in the control plane (e.g., NLSR \cite{nlsr}  or DCR \cite{dcr}).

Algorithm \ref{algo-CCN-RAMP-Interest} shows the steps taken by a router to process an Interest  received from a local consumer or a neighbor router, which were discussed in Section 4.6.
An Interest from a consumer is assumed to specify the name of a requested CO with the rest of the information being nil.

\begin{algorithm}[h]
\caption{Processing Interest  from $p$ }
\label{algo-CCN-RAMP-Interest}
 {\fontsize{6.5}{6.5}\selectfont
\begin{algorithmic}
\STATE{{\bf function} Interest\_Forwarding}
\STATE {\textbf{INPUT:} $LIST^i$,$FAB^i$, $PRT^i$, $LSAT^i$; \\
\textbf{INPUT:}
$I[n(j),  AID^I(p), anchor, D^I(p)]$;}
\IF{$p[n(j)] \in CS$ }
		\STATE{
		retrieve CO $n(j)$;   
		send  $DP[n(j), AID^R(i), sp(j) ]$ to $p$}
\ELSE
	\IF{p is consumer}
			\STATE{$lc[n(j)]  = lc[n(j)]  \cup c;$ }
			\FOR{{\bf each} $a$ {\bf by rank in} $PRT^i(n(j))$} 
				\STATE{
				$anchor=a;$ 
				break;}

			\ENDFOR	
			\STATE{$aid=f(anchor);$}
	\ELSE
		\STATE{$aid=AID^I(p);$}
	\ENDIF

	\STATE {$entry = nil;$}
	\FOR{{\bf each} $e \in LSAT^i_{aid}(aid)$ }
		\IF{$PH(e) = p $ }
			\STATE{$entry=e;$ $break;$ }
		\ENDIF
	\ENDFOR

	\STATE {$ loop=true; noRoute=true;$}
	\IF {$entry = nil$}
		\FOR{{\bf each} $s \in N^i$ {\bf by rank in} $FAB^i(anchor)$} 
			\STATE{NoRoute=false;}
			\IF {$ D^I(p)  > D(i, n(j)^*, s) $   ~(\% ALF is satisfied) }
				\STATE{
			loop=false;
			   $D^I(i) = D(i, n(j)^*, s)$; 
				$NH=s$;\\
						  break;	}
			\ENDIF			
		\ENDFOR	
		\STATE{
		select unused random number   $map$; \\
			$entry = $create entry $LSAT^i[aid, p, map, NH, D(i, n(j)^*, s)];$
		}
	\ENDIF

\IF{$entry \neq nil$ }
\STATE{
			send $I[n(j), MAP(entry), anchor, D(entry)]$ to $NH(entry)$; }
			{\bf return;}
\ENDIF
	\IF{$noRoute=false \land loop = true$}
		\STATE{
			send $ERR[n(j),  AID^I(p), anchor, \mathsf{loop} ]$ to $p$ ;
		}
	\ELSE
		\STATE{
		 send $ERR[n(j),  AID^I(p), anchor, \mathsf{no~ route} ]$ to $p$ ~\\ (\% No route to $n(j)^*$ exists);
		}
	\ENDIF

\ENDIF

\end{algorithmic}
}
\end{algorithm}

Algorithm \ref{algo-CCN-RAMP-Data} shows the steps taken 
when a  data packet  from router $s$ is received at router $i$. Like an interest, a data packet contains an anonymous identifier $AID^R(s)$. Router looks up  $LSAT^(i)$ for $AID^R(i)$. 
If   no entry with $MAP=AID^R(i)$ exists, the router does not forward the packet any further. If a matching entry is found, the router checks if the previous hop stated in the 
$LSAT^(i)$ entry  is a neighbor router or the router itself. If it is a neighbor node, it forwards the packet to the previous hop $PH$ stated in the matched $entry$. Otherwise, the data packet is forwarded to the local consumers listed in $lc[n(j)]$ of the entry for  $n(j)$ in $LRT^i$.

Algorithm \ref{algo-link-failure} shows the steps taken when the link connecting router $i$ to router $s$ fails. In such a case, for each entry in $LSAT^i$ that states the next hop as router $s$, an error message including the $AID$ of the $entry$ is created and is sent back to the previous hop stated in the AID entry. This way, router $i$ informs neighbor routers of the link failure. Router $i$ also invalidates all the matching entries in $LSAT^i$.

Algorithm \ref{algo-CCN-RAMP-ERR} shows the steps followed after
an error message from router $s$ is received at router $i$. The received error message contains an $AID$ set by the neighbor router. Router $i$  looks up $LSAT^i$ and for any entry with $MAP=AID^R(s)$, it  creates an error message containing $AID(entry)$ and sends it to the previous hops or the local consumers. The router also invalidates the  entry found.

\vspace{-0.05in}
\begin{algorithm}[h]
\caption{Processing data packet from router $s$ }
\label{algo-CCN-RAMP-Data}
{\fontsize{7}{7}\selectfont
\begin{algorithmic}
\STATE{{\bf function} Data Packet}
\STATE{\textbf{INPUT:}  $LIST^i$, $LST^i$, 
$DP[n(j),  AID^R(s), sp(j) ]$; }
\STATE{{\bf [o]} verify $ sp(j)$;}
\STATE{{\bf [o]} {\bf if} verification with $ sp(j)$ fails {\bf then} \\
~~~~~discard $DP[n(j), AID^R(s), sp(j) ]$;}

\STATE {
$entry = LSAT^i_{map}(AID^R(s));$~( \%LST in case of CCN-RAMP)
}
\IF {$entry=nil$}
	\STATE{$drop;$ 
	$return;$}
\ENDIF

\IF{ $preHop(entry) = local$~~(\% router $i$ is the origin) }
	\FOR{{\bf each} $c \in lc[n(j)]$}
		\STATE{send $DP[n(j),nil,  sp(j) ]$ to  $c$; $lc[n(j)] = lc[n(j)] - \{ c \}$}
	\ENDFOR
\ELSE
	\STATE{
	send $DP[n(j), AID(entry), sp(j)]$ to  $preHop(entry)$;
	}
\ENDIF

store CO in CS

\end{algorithmic}}       
\end{algorithm}

\vspace{-0.15in}
\begin{algorithm}[h]
\caption{Failure of link $l$ connected to router i to p}
\label{algo-link-failure}
{\fontsize{7}{7}\selectfont
\begin{algorithmic}
\STATE{{\bf function} Data Packet}
\STATE{\textbf{INPUT:} $LSAT^i$; }
\FOR{{\bf each} $entry \in LSAT^i$ with $NextHop(entry)=p$}
	\STATE{$preHop = getNodeID( AID(entry));$}
	\STATE{send $ERR[nil,  AID(entry), \mathsf{link~ failure} ]$ to $preHop$}
	\STATE{INVALIDATE(entry)};
\ENDFOR
\end{algorithmic}}       
\end{algorithm}

\vspace{-0.15in}
\begin{algorithm}[h]
\caption{Processing Error Message from Router $s$ at router $i$}
\label{algo-CCN-RAMP-ERR}
{\fontsize{7}{7}\selectfont
\begin{algorithmic}
\STATE{{\bf function} ERR}
\STATE{\textbf{INPUT:} $LSAT^i, ERR[n(j),  AID^R(s), reason ]$; }
\FOR{{\bf each} $entry \in LSAT^i_{map}(AID^R(s))$}
	\STATE{$preHop = getNodeID( AID(entry));$}
	\IF{ $preHop = local$~~(\% router $i$ was the origin of the Interest) }
		\FOR{{\bf each} $c \in lc[n(j)]$}
			\STATE{send $DP[n(j),c,  sp(j) ]$ to  $c$; $lc[n(j)] = lc[n(j)] - \{ c \}$}
		\ENDFOR
	\ELSE
		\STATE{send $ERR[nil,  AID(entry), \mathsf{reason} ]$ to $preHop$;}
	\ENDIF
	\STATE{INVALIDATE(entry);}
\ENDFOR
\end{algorithmic}}       
\end{algorithm}

 \vspace{-0.07in}
\section{Performance Comparison}
\label{sec-perf}

We implemented CCN-RAMP in ndnSIM \cite{ndnsim}  based on Algorithms 1 to 4 and used the NDN implementation in ndnSIM without modifications to compare NDN with CCN-RAMP.   We also compare CCN-RAMP against our previous proposal for the elimination of PITs using datagrams,  CCN-GRAM \cite{ifip2016}, which also relies on name-based routing in the control plane and  uses 
FIBs listing the next hops to known name prefixes to forward Interests.

The performance metrics used for comparison 
are  the average sizes of forwarding tables,  the average number of table lookups needed to obtain one CO,  the average end-to-end delays, and the average number of Interests  sent by routers.  
DCR \cite{dcr} is used in the control plane  to update  FIBs listing name prefixes
in NDN and CCN-GRAM, or the  Prefix Resolution Tables (PRT) and Forwarding to Anchor Bases (FAB) used in CCN-RAMP. Accordingly, we do not  need to consider the signaling overhead of the name-based routing protocol, because it is exactly the same in the three approaches we consider.

We considered networks with no caching and with
on-path caching, with each cache being able to store only 1000 COs.
We used the AT\&T network topology, which is considered to be a realistic topology for simulations  \cite{att}. This topology includes 153  nodes and 184 point-to-point links with 30 ms delay.  To reduce the effects derived from sub-optimal implementations of CCN-RAMP, NDN, or CCN-GRAM, we set the data rate of point-to-point links to 10Gbps. 

We selected 70  nodes randomly to have a consumer application simulating local consumers.
All consumers generate Interests requesting COs from all name prefixes
following  a Zipf distribution with parameter $\alpha=0.7$.
The total number of COs is only  $10^7$, with  1000 COs per name prefix.
We used a relatively small content population, because using larger numbers for COs and name prefixes would just make CCN-RAMP look better compared to NDN.

We selected 
20 of the nodes randomly   to be anchors of 500 different name prefixes each,
and each name prefix has a single anchor. 

Given that prefixes are not multi-homed in our simulation scenarios, the results in Section 3 indicate that the paths traversed by Interests should be the same whether forwarding tables maintain routes to anchors or to prefixes. This is the case  for both single-path and multi-path routing. Accordingly, for simplicity, our simulation experiments assume single-path routing and static topologies. The difference between CCN-RAMP and CCN-GRAM \cite{ifip2016} (its counter-part based on FIBs listing name prefixes) is that the latter incurs more forwarding overhead by its use of FIBs.

\subsection{Average Table Sizes}

Figure \ref{TableSize} shows the average table sizes for NDN, CCN-RAMP, and CCN-GRAM  on a logarithmic scale
as a function of the rate at which Interests arrive at routers with local consumers,
ranging from 100 to 2000 Interests per second.

\vspace{-0.10in}
\begin{figure}[h]
\begin{centering}
    \mbox{
    \subfigure{\scalebox{.6}{\includegraphics{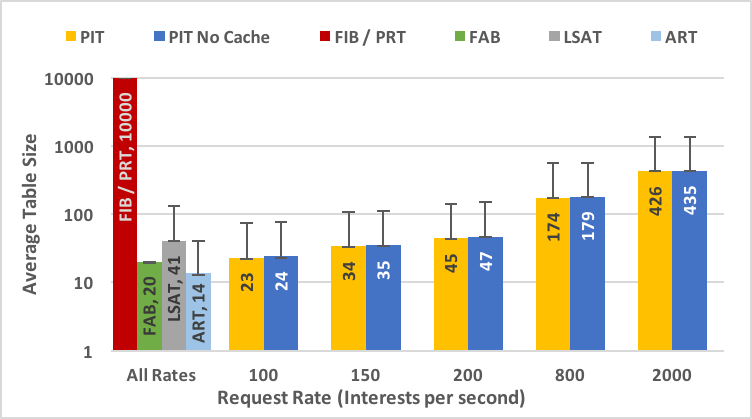}}}
      }
\vspace{-0.16in}
   \caption{Average number of entries in forwarding tables  for NDN, CCN-GRAM,  and CCN-RAMP }
   \label{TableSize}
\end{centering} 
\end{figure}  

The  number of entries in the PRTs used in CCN-RAMP is the same as the 
number of entries in the FIBs used in NDN and CCN-GRAM, because both list an entry for each known name prefix.
Given that the topology contains 20 producer nodes and each node has 500 different name prefixes, each PRT and FIB table is expected to have  10,000 entries. 

For CCN-RAMP,  the average number of FAB entries for all Interest rates is only 20 and the average size of an 
LSAT, which is used to forward responses to Interests back to consumers, is only 41 entries for all values of Interest rates.
The small sizes of FABs and LSATs should be expected, because there are only 20 routers acting as anchors, and each router acts as a relay of only a fraction of the paths to such anchors. 

The average size of the forwarding table used in CCN-GRAM to send responses to Interests towards consumers (called ART) is only 14 entries for all Interest rates.

By contrast, the number of PIT entries depends on network conditions and traffic load. 
The average PIT size varies from 23 or 24 to 426 or 435, depending on 
whether on-path caching is used.  
The size of PITs at some core routers can be more than 1000 entries 
when the request rate at routers with local consumers is  2000 Interests per second. Interestingly, the average number of  PIT entries is not much smaller when on-path caching is used compared to the case in which no caching is used.

\subsection{Average Number of Table Lookups}

Figure \ref{lookupCnt} shows the average number of table lookups  required to retrieve a single CO, and includes all lookups done in forwarding the Interest for the CO and sending back the corresponding CO to the requesting origin router. 

\vspace{-0.10in}
\begin{figure}[h]
\begin{centering}
    \mbox{
    \subfigure{\scalebox{.63}{\includegraphics{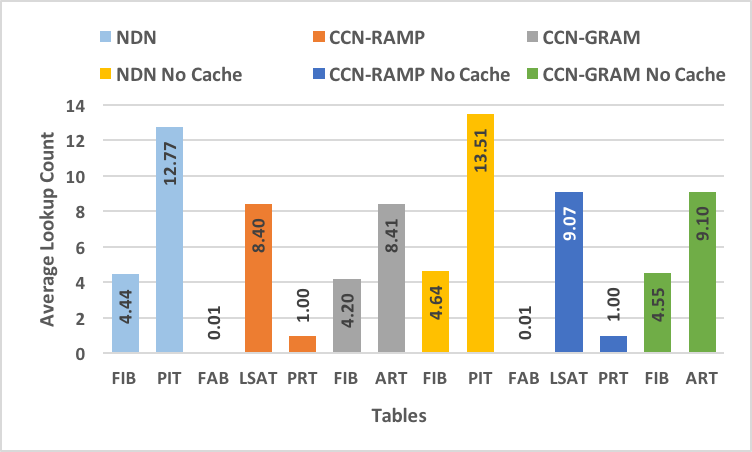}}}
      }
\vspace{-0.1in}
   \caption{Average number of table lookups needed to retrieve one CO}
   \label{lookupCnt}
\end{centering} 
\end{figure}   

In NDN, forwarding an Interest requires  a PIT lookup at each hop along the path from consumer to caching site or anchor. If no PIT entry is found, a FIB lookup is required to obtain the next hop for the Interest. Forwarding a data packet sent from an anchor or a caching site requires a PIT lookup at every hop 
along the way to the consumer. 

In CCN-RAMP, retrieving a remote CO includes one PRT lookup at the router that receives the Interest from a local consumer. That router binds the CO name to an anchor and each router along the path towards the anchor 
must do one FAB lookup to forward the first Interest with an AID that does not exist in the LSAT of the router, and one LSAT lookup for every Interest being forwarded.  Once forwarding state is established along a path from an originating router to an anchor, no FAB lookups are needed for Interests that carry AIDs already listed in the LSATs of the relaying routers. 
As Figure \ref{lookupCnt} shows, the average number of FAB lookups per Interest is a very small fraction, because only a small fraction of Interests are forwarded without having any forwarding state already established in the LSATs of routers.  

In CCN-GRAM, retrieving a remote CO involves one FIB lookup at each hop  along the path from consumer to caching site or anchor, as well as a lookup of the ART in order to carry out the proper swapping of AIDs.  Forwarding a data packet sent from an anchor or a caching site requires only  an ART lookup at each  hop along the way to the consumer. 

Compared to NDN and CCN-GRAM, CCN-RAMP results in relay routers doing  fewer  lookups  of tables that are three orders of magnitude smaller than FIBs listing name prefixes, even for the small scenario we consider.

It can be inferred from  Figure \ref{lookupCnt} that the average hop count for paths traversed by Interests is  around four or five hops, because this is the approximate number of average FIB lookups needed in NDN  and CCN-GRAM when no caching is used.
As should be expected,  the average number of table lookups in NDN, CCN-GRAM, and CCN-RAMP is slightly larger when no on-path caching is available, because the average paths between consumers requesting COs and the anchors are longer than the average paths between consumers and caching sites.

\vspace{-0.10in}
\begin{figure}[h]
\begin{centering}
    \mbox{
    \subfigure{\scalebox{.64}{\includegraphics{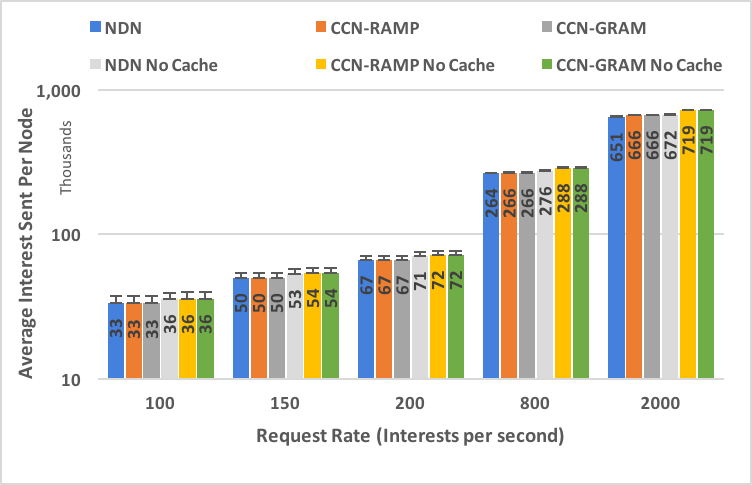}}}
      }
\vspace{-0.1in}
   \caption{Average number of Interests forwarded per router}
   \label{intCnt}
\end{centering} 
\end{figure}  

\vspace{-0.24in}
\begin{figure}[h]
\begin{centering}
    \mbox{
    \subfigure{\scalebox{.62}{\includegraphics{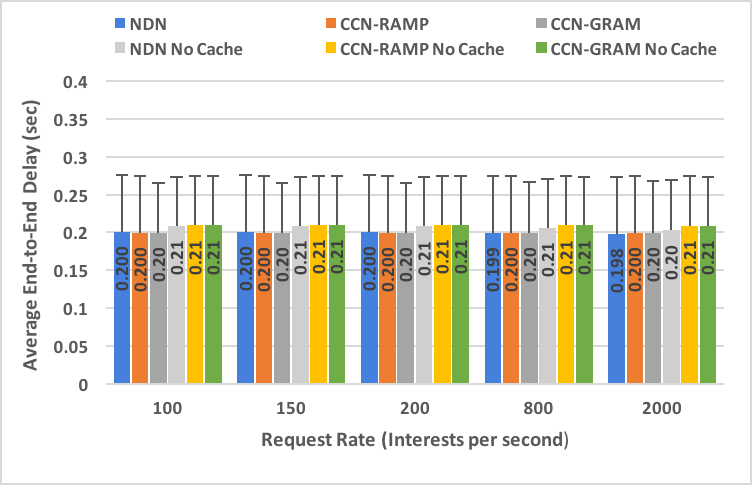}}}
      }
\vspace{-0.1in}
   \caption{Average end-to-end delays }
   \label{delay}
\end{centering} 
\end{figure}

\subsection{Average Number of Interests and \\ End-to-End Delays}

Figure \ref{intCnt} shows the average number of Interests sent by each router. 
As the results indicate,  the average encumber of Interests sent by each router is essentially the same for all three approaches and the percentage of Interests that benefit from aggregation using PITs is insignificant. 

Figure  \ref{delay} shows that  the average end-to-end delays are  very similar in all cases for NDN and CCN-RAMP. 
Given that the simulations assume zero delays for table lookups (i.e., the differences in forwarding table sizes are not taken into account), these simulation results  indicate that the paths traversed by Interests are the same for NDN, CCN-GRAM,  and CCN-RAMP. 
This confirms that  the paths traversed by  Interests
when  routers maintain FIBs with entries for name prefixes tend to be  the same as the paths traversed by Interests  if the origin routers select the anchors of name prefixes and routers forward Interests towards anchors.

\section{Conclusions}

Scaling has been identified as a major research problem for content-centric networks \cite{forest}.
We  introduced CCN-RAMP, 
a new  approach to  content-centric networking that can be deployed at Internet scale and is able to handle billions of name prefixes, because it eliminates the 
need to lookup large FIBs listing name prefixes, and  the 
use of  PITs that make routers vulnerable to DDoS attacks on the routing infrastructure.
CCN-RAMP provides all the benefits sought by NDN and CCNx, including native support of multicasting without the need for a new multicast routing protocol (see  \cite{ifip2016, globe2016}). In contrast to NDN and CCNx, Interests cannot traverse forwarding loops and no Interest-flooding attacks can be mounted.

The results of simulation experiments based on implementations of NDN, CCN-GRAM,  and CCN-RAMP in ndnSIM  show that CCN-RAMP is more efficient than NDN and CCN-GRAM.   CCN-RAMP rendered similar  end-to-end delays, incurred similar Interest overhead in the data plane, and resulted in  forwarding state with a number of entries  that can be orders of magnitude smaller than the forwarding state required in NDN.

Our results open up several research avenues on content-centric networking at Internet scale. Important next steps include: analyzing the performance impact of the dynamics of  name-based routing protocols adapting to congestion and 
topology changes, analyzing the impact of multihoming of name prefixes, enabling hierarchical name-based routing in the control plane, defining the role of autonomous systems, and analyzing policy-based
routing and forwarding across autonomous systems.

\vspace{-0.05in}
 {\fontsize{7}{7}\selectfont

  }

\end{document}